\title{Exponential Weight Functions for Quasi-Proportional Auctions}
\author{Eric Bax \and James Li \and Zheng Wen\\
Yahoo Labs}
\date{\today}

\documentclass[12pt]{article}

\usepackage{amssymb}
\usepackage{amsmath}
\usepackage{subfigure}
\usepackage{graphicx}

\newtheorem{theorem}{Theorem}[section]

\newtheorem{corollary}[theorem]{Corollary}

\newenvironment{proof}[1][Proof]{\begin{trivlist}
\item[\hskip \labelsep {\bfseries #1}]}{\end{trivlist}}

\newcommand{\qed}{\nobreak \ifvmode \relax \else
      \ifdim\lastskip<1.5em \hskip-\lastskip
      \hskip1.5em plus0em minus0.5em \fi \nobreak
      \vrule height0.75em width0.5em depth0.25em\fi}

\newcommand{\be}{\begin{equation}}
\newcommand{\ee}{\end{equation}}

\newcommand{\setwv}{[w_1, v_1] \times \ldots \times [w_n, v_n]}

\newcommand{\bb}{\mathbf{b}}
\newcommand{\cc}{\mathbf{c}}
\newcommand{\ww}{\mathbf{w}}
\newcommand{\br}{\mathrm{BR}}
\newcommand{\brbr}{\mathrm{\mathbf{B}\mathbf{R}}}

\begin{document}
\maketitle

\begin{abstract}
In quasi-proportional auctions, the allocation is shared among bidders in proportion to their weighted bids. The auctioneer selects a bid weight function, and bidders know the weight function when they bid. In this note, we analyze how weight functions that are exponential in the bid affect bidder behavior. We show that exponential weight functions have a pure-strategy Nash equilibrium, we characterize bids at an equilibrium, and we compare it to an equilibrium for power weight functions. 
\end{abstract}

\section{Introduction}
\label{sec_introduction}

Quasi-proportional auctions \cite{tullock80,kelly97} award each bidder a share of the allocation proportional to their weighted bid. Specifically, if $\bb$ is the vector of bids, $b_i$ is bidder $i$'s bid, and $f$ is the weight function, then bidder $i$ has allocation:
\be
a_i(\bb) = \frac{f(b_i)}{\sum_j f(b_j)}.
\ee
We focus on winners-pay, so each bidder pays their bid times their allocation. Let $v_i$ be bidder $i$'s private value for a full allocation and assume linear valuation. Then bidder $i$ has utility (profit):
\be
u_i(\bb) = (v_i - b_i) a_i(\bb).
\ee

For \textit{power weight functions} of the form $f(x)=x^p$ with $p>0$, Wen et al. \cite{wen15} show that a pure-strategy Nash equilibrium exists and supply lower bounds for bids at an equilibrium. For \textit{exponential weight functions} of the form $f(x)=e^{cx}-1$, we show that a pure-strategy Nash equilibrium exists, characterize bids at an equilibrium, and give lower bounds for those bids. In addition, we compare exponential and power weight functions, showing that the revenue-maximizing exponential weight function produces more revenue than the revenue-maximizing power weight function for even moderately competitive auctions. 

A second-price winner-take-all auction with reserve prices is known to be revenue-optimal \cite{myerson81,riley81}. So why explore revenue maximization for quasi-proportional auctions? First, the requirements for setting optimal reserve prices are not always met in practice. Second, the quasi-proportional auction has some properties that make it preferable in some settings.

The second-price winner-take-all auction with reserve prices maximizes expected revenue given prior distributions for bidders' private values. There can be problems if the auctioneer's priors for bidders' private values are unknown, innaccurate, or very precise. Unknown priors leave the auctioneer in the \textit{prior-free} setting \cite{goldberg02,hartline07}. In some cases, priors may be learned from bids in successive auctions \cite{li10,cole14,dughmi14,hummel14}, in others, private values or their distributions may change too much from auction to auction for learning to be very effective. Inaccurate priors produce reserve prices that are not optimal. This can significantly reduce revenue, because expected revenue can be very sensitive to reserve prices \cite{muthukrishnan09}. 


Quasi-proportional auctions have the following properties that can be advantages over the second-price winner-take-all auction with reserve prices:
\begin{enumerate}
\item \textbf{Full Allocation.} Quasi-proportional auctions always result in sales, assuming at least one bidder with a nonzero bid. With reserve prices, there is no sale unless some bidder bids at least their reserve price. 
\item \textbf{Shared Allocation.} Quasi-proportional auctions award some allocation to each bidder who enters a positive bid. This ``second prize for second price" \cite{wen15} can give bidders an incentive to bid even if they know they are unlikely to be the highest bidders. 
\item \textbf{Smooth Response to Bid Changes.} In a quasi-proportional auction, each bidder's allocation and payment varies continuously in their bid, increasing for any bid increase and decreasing for any bid decrease. With winner-take-all, a slight increase or decrease in bid either has no effect or completely alters the allocation and payment.
\item \textbf{Symmetric Outcomes for Bidders.} In a quasi-proportional auction, if two bidders swap bids, then they swap allocations and payments. With reserve prices, this may not hold, because bidders with different priors will, in general, have different reserve prices. 
\end{enumerate}

\section{Pure-Strategy Nash Equilibria}
In this section, we show that for $f(x)=e^{cx}-1$ there is a pure-strategy Nash equilibrium, we characterize bids at an equilibrium, and we give lower bounds for bids at an equilibrium. To simplify notation, when we focus on a single bidder, we drop the bidder's subscript (for example using $b$ for the bid instead of $b_i$), we drop the arguments in parentheses for functions, and we use apostrophes to denote derivatives with respect to $b$. For example, $f''$ is the second derivative of the weight function with respect to bid. 

Focus on a single bidder. Their \textit{response curve} is their utility as a function of their bid, given other bidders' bids. Their \textit{best response} is the bid that maximizes their utility given other bidders' bids. 

\begin{theorem} \label{unique}
With weight function $f(x)=e^{cx}-1$, for any $\bb$ with at least two positive bids, each bidder's response curve has a single extremum over $[0,v]$, and it is a maximum. 
\end{theorem}

\begin{proof}
We require at least two positive bids to ensure that no bidder's response curve has a zero denominator for the allocation. (With a single positive bid, that bidder's response curve would have a zero denominator at zero.) Note that utility is zero at bids 0 and $v$, positive in between, bounded, and continuous. So it has a maximum in $(0,v)$. 

Lemma 5 in the proof of Theorem 1 in \cite{wen15} states that if $f f'' < 2(f')^2$ then $(u'=0) \implies (u''<0)$. For us, $f = e^{cx}-1$, $f'=ce^{cx}$, and $f''=c^2 e^{cx}$. Substitute and simplify:
\be
f f'' < 2(f')^2.
\ee
\be
\left(e^{cx}-1\right) c^2 e^{cx} < 2 c^2 e^{2cx}.
\ee
\be
e^{cx}-1 < 2 e^{cx}.
\ee
\be
-1 < e^{cx}.
\ee
So $(u'=0) \implies (u''<0)$, meaning that every extremum is a local maximum, so there are no local minima in $(0,v)$. If there were multiple local maxima, then each successive pair would have a local minimum in between. So the lack of interior local minima implies at most one maximum. \qed
\end{proof}

Let $\br_i(\bb)$ be bidder $i$'s best response to bids $b_1, \ldots, b_{i-1}, b_{i+1}, \ldots, b_n$ from $\bb$, and let $\brbr$ be the vector of best responses $(\br_1(\bb), \ldots, \br_n(\bb))$. At pure-strategy Nash equilibrium bids $\bb^*$, $\brbr(\bb^*) = \bb^*$. Now we characterize bids at a pure-strategy Nash equilibrium. (Later we prove one exists.) 

\begin{theorem} \label{char}
At any pure-strategy Nash equilibrium $\bb^*$ with at least two positive bids,
\be
\forall i: b^*_i = v_i - \frac{1}{c} \left( 1 - \frac{1}{e^{cb^*_i}} \right) \left( \frac{1}{1-a_i(\bb^*)} \right).
\ee
\end{theorem}

\begin{proof}
Focus on a single bidder. Since 
$$ u = (v-b)a, $$
$$ u' = -a + (v-b)a'. $$
Since 
$$ a = \frac{f}{f+s}, $$
where $s = \sum_{j \not= i} f(b_j)$ is the sum of other bidders' weighted bids,
$$ a' = \frac{f'}{f+s} - \frac{f f'}{(f+s)^2} $$
$$ = \frac{f'}{f+s} (1-a). $$
So 
$$ u' = -a + (v-b)(1-a)\frac{f'}{f+s}. $$
Set $u'=0$ and solve for $b$ to find the best response.
$$ b = v - \frac{a}{1-a} \frac{(f+s)}{f'}. $$
Since $a = \frac{f}{f+s}$,
$$ b = v - \frac{f}{f'} \left( \frac{1}{1-a} \right). $$ \label{eq_char}
Since $f = e^{cb}-1$ and $f' = c e^{cb}$,
$$ b = v - \frac{1}{c} \left( 1 - \frac{1}{e^{cb}} \right) \left( \frac{1}{1-a} \right). \qed $$
\end{proof}

This characterization is not a closed-form solution, because $e^{cb^*_i}$ depends on $b^*_i$ and $a_i$ depends on $\bb^*$. However, it will give some insight about equilibrium, and we will extend this theorem to get bounds for bids at equilibrium and prove that one exists. 

For some insight on equilibrium for exponential weight functions, we will compare to an equilibrium characterization for power weight functions. Substitute $f=b^p$ and $f'=p b^{p-1}$ into Equality \ref{eq_char}. 
$$ b = v - \frac{b}{p} \left( \frac{1}{1-a} \right). $$
Solve for $b$:
$$ b = \frac{v}{1 + \frac{1}{p} \left( \frac{1}{1-a} \right)}. $$
So for any pure-strategy Nash equilibrium $\bb^*$ with at least two positive bids, 
$$ \forall i: b^*_i = \frac{v_i}{1 + \frac{1}{p} \left( \frac{1}{1-a_i(\bb^*)} \right)}. $$

Compare this to Theorem \ref{char}. At a high level, equilibrium bids for exponential weight functions subtract an amount from the private value, but equilibrium bids for power weight functions divide the private value by one plus an amount. The amounts have similar forms. They share the term $\frac{1}{1-a}$, which decreases equilibrium bids as the bidder's share of the equilibrium allocation increases. The terms $\frac{1}{c}$ and $\frac{1}{p}$ increase equilibrium bids as steeper weight functions are selected. For both exponential and power weight functions, the steepness parameter ($c$ or $p$) mediates a tradeoff: a steeper weight function increases bids via $\frac{1}{c}$ or $\frac{1}{p}$ but decreases the highest bidder's equilibrium bid by increasing their share of the allocation, which increases $\frac{1}{1-a}$. Balancing these effects maximizes revenue. For both exponential and power weight functions, increasing competition -- by having closer private values among top bidders, by having more bidders, or both -- allows the auctioneer to increase the steepness parameter without allowing the bidder with the highest private value to submit a low bid relative to their private value and still capture the lion's share of the allocation. As a result, increasing competition increases the steepness of the revenue-maximizing weight function and increases the equilibrium revenue for that weight function.

The next theorem characterizes lower bounds $\ww$ for bids at an equilibrium.

\begin{theorem} \label{bound}
For $f(x)=e^{cx}-1$, if $\ww$ satisfies:
\be
\forall i: w_i \leq v_i - \frac{1}{c} \left(1 - \frac{1}{e^{cw_i}}\right)\left(\frac{1}{1-a_i(\ww)}\right),
\ee
where $a_i(\ww)$ is the allocation to bidder $i$ if $\bb = \ww$, then $(\forall i: b_i \geq w_i) \implies (\forall i: \br_i(\bb) \geq w_i)$.
\end{theorem}

\begin{proof}
Focus on a single bidder $i$. Assume $\forall j \not= i: b_j \geq w_j$. Let $u$ be bidder $i$'s utility function given other bidders' bids, and let $b$ be bidder $i$'s response bid. By Theorem \ref{unique}, $u$ has a single local maximum, $u' < 0$ before the maximum, and $u' > 0$ after the maximum. So if $u' \geq 0$ at $b=w$, then the best response is $w$ or greater. 

From the proof of Theorem \ref{char}, solve for $u'\geq0$ at $b=w$ rather than $u'=0$ at $b$:
\be
w \leq v - \frac{1}{c} \left(1 - \frac{1}{e^{cw}}\right) \left(\frac{1}{1-a(w,\bb)}\right),
\ee
where $a(w,\bb)$ is the allocation to bidder $i$ if bidder $i$ bids $w$ and the other bidders maintain their bids from vector $\bb$. Let $a(\ww)$ be the allocation to bidder $i$ if $\bb=\ww$. Since we assume $\forall j \not= i: b_j \geq w_j$, $a_i(\ww) \geq a(w,\bb)$. (Decreasing competitors' bids increases bidder $i$'s allocation.) So
\be
v - \frac{1}{c} \left(1 - \frac{1}{e^{cw}}\right) \left(\frac{1}{1-a(\ww)}\right) \leq v - \frac{1}{c} \left(1 - \frac{1}{e^{cw}}\right) \left(\frac{1}{1-a(w,\bb)}\right),
\ee
and 
\be
w \leq v - \frac{1}{c} \left(1 - \frac{1}{e^{cw}}\right) \left(\frac{1}{1-a(\ww)}\right). \qed
\ee
\end{proof}

Now we show that a pure-strategy Nash equilibrium exists.

\begin{theorem} \label{exist}
Suppose $\ww$ satisfies:
$$ \forall i: w_i \leq v_i - \frac{1}{c} \left(1 - \frac{1}{e^{cw_i}}\right)\left(\frac{1}{1-a_i(\ww)}\right). $$
Then 
$$ \exists \bb^* \in \setwv: \brbr(\bb^*) = \bb^*. $$
\end{theorem}

\begin{proof}
According to Brouwer's fixed point theorem \cite{brouwer12,franklin02}, a function has a fixed point in a convex compact set if the function is continuous and maps the set into itself. Theorem \ref{bound} shows that the function $\brbr$ maps $\setwv$ into itself. \qed
\end{proof}

Having shown that there is an equilibrium, we can use Theorem \ref{bound} to derive bounds for bids at that equilibrium. For example,

\begin{corollary}
Without loss of generality, let $v_1 \geq \ldots \geq v_n$. Let
$$ w_1 = v_2 - \frac{2}{c}, $$
and 
$$ \forall i \geq 2: w_i = v_i - \frac{2}{c}. $$
Then $(\forall i: b_i \geq w_i) \implies (\forall i: \br_i(\bb) \geq w_i)$. 
\end{corollary}

\begin{proof}
We will show that the $\ww$ in the corollary meets the conditions of Theorem \ref{bound}. For $w_1$, since $w_1=w_2$, $a_1(\ww) = a_2(\ww)$, so $a_1(\ww) \leq \frac{1}{2}.$ So
$$ w_1 = v_2 - \frac{2}{c} \leq v_1 - \frac{1}{c} \left( \frac{1}{1 - a_1(\ww)} \right) \leq v_1 - \frac{1}{c} \left(1 - \frac{1}{e^{cw_1}}\right)\left(\frac{1}{1-a_1(\ww)}\right). $$
For $i \geq 2$, $w_i \leq w_{i-1}$, so $a_i(\ww) \leq \frac{1}{2}$. Hence
$$ w_i = v_i - \frac{2}{c} \leq v_i - \frac{1}{c} \left( \frac{1}{1 - a_i(\ww)} \right) \leq v_1 - \frac{1}{c} \left(1 - \frac{1}{e^{cw_i}}\right)\left(\frac{1}{1-a_i(\ww)}\right). \qed $$
\end{proof}

\section{Tests}
This section presents test results showing that exponential weight functions produce more revenue than power weight functions when there is even a modest level of competition among bidders. The tests include a two-bidder scenario ($n = 2$) and a ten-bidder scenario ($n = 10$.) For the ten-bidder scenario, there is a single bidder with a higher private valuation, and the others have equal private valuations. This scenario is called OLOS (one large and others the same) in \cite{wen15}. 

For each scenario, we set $v_1 \geq v_2 = \ldots = v_n$, and define $\alpha = \frac{v_1}{v_2}$. For each $\alpha \in \left\{1.2, 1.4, \ldots, 10.0 \right\} \cup \left\{20, 30, \ldots, 100 \right\}$, we ran tests for a range of $c$ and $p$ values to determine which values of these steepness parameters maximized auction revenue for exponential and power weight functions, respectively. In each test, we set $v_1 = \alpha$ and $v_2 = \ldots = v_n = 1$, initialized all bids to $\frac{1}{2}$, then for 100 iterations used the best response for each bidder to the other bidders' previous bids as the bidder's next bid. In other words, we set $\bb^0 = (\frac{1}{2}, \ldots, \frac{1}{2})$, then repeated $\bb^{i+1} = \brbr(\bb^i)$ one hundred times to get final bids $\bb^{100}$. We used the final bids to determine revenue for the scenario and $c$ or $p$ value. 

Figure \ref{rev_steep} displays the maximum revenue and revenue-maximizing $c$ and $p$ values for the two-bidder scenario. Figures \ref{sa_rev} and \ref{ma_rev} show that the best exponential weight function produces more revenue than the best power weight function for $\alpha < 70$. In other words, with two bidders, the best exponential weight function generates more auction revenue than the best power weight function unless one bidder's private value is at least 70 times the other's. Figures \ref{sa_steep} and \ref{ma_steep} show that very steep weight functions maximize revenue when bidders' private values are nearly equal, but the revenue-maximizing parameter values decrease as competition decreases. In other words, when there is stiff competition, the auctioneer can run something close to a winner-take-all auction using a steep weight function, but when there is weak competition, the auctioneer must use a flatter weight function to threaten the high-value bidder with having to share the allocation with the low-value bidder. For strong or moderate competition, exponential weight functions provide something closer to a winner-take-all auction because of their steeper shape, while flatter power weight functions provide a stronger threat to share the allocation almost equally among bidders when there is very weak competition. 

Figure \ref{alloc_bids} shows the allocation to the high-value bidder and the bids for the revenue-maximizing $c$ and $p$ values with $n = 2$. Exponential and power weight functions award similar allocations for strong competition (low $\alpha$), but exponential weight functions allocate more to the high bidder as competition decreases. This is consistent with power weight functions threatening, and partially achieving, more sharing under weak competition. Exponential weight functions result in the high and low bidder having closer bids under strong competition, but this reverses as competition weakens, with equally close bids at about $\alpha = 50$. Steeper weight functions force the bidder with lower private value to bid higher in order to get more than a very small portion of the allocation; this, in turn, can drive the bidder with higher private value to bid higher. As weight functions get less steep (with increasing $\alpha$), the lower bids actually decrease, because the lower bidder can obtain a sizable portion of the allocation even with a low bid. 

Figure \ref{olos_rev_steep} shows revenue and optimal steepness results for the $n = 10$ scenario, which has more competition than the $n = 2$ scenario, in the form of more bidders. Figures \ref{osa_rev} and \ref{oma_rev} show that, under this increased competition, the difference between exponential and power weight function revenue increases, and exponential weight functions continue to produce more revenue than power weight functions even for $\alpha = 100$. (The crossover is between 100 and 200 -- not shown on these figures.) Figures \ref{osa_steep} and \ref{oma_steep} show that optimal weight functions are flatter as competition decreases (as $\alpha$ increases). 

Figure \ref{olos_alloc_bids} shows allocations and bids for revenue-maximizing weight functions with $n=10$. Figures \ref{osa_alloc} and \ref{oma_alloc} show that the allocation to the highest bidder decreases gradually as $\alpha$ increases for power weight functions. For exponential weight functions, the allocation is close to the allocation for power weight functions until $\alpha>10$. Then the allocation for power weight functions decreases as $\alpha$ increases. In contrast, for exponential weight functions, the allocation increases, giving more allocation to the highest bidder as competition weakens. Figures \ref{osa_bids} and \ref{oma_bids} show that, as for $n=2$, exponential weight functions produce more of a spread between bids under more competition, and this reverses for $\alpha > 70$.

\begin{figure}[ht]
\subfigure[]{\includegraphics[width = 0.48\textwidth]{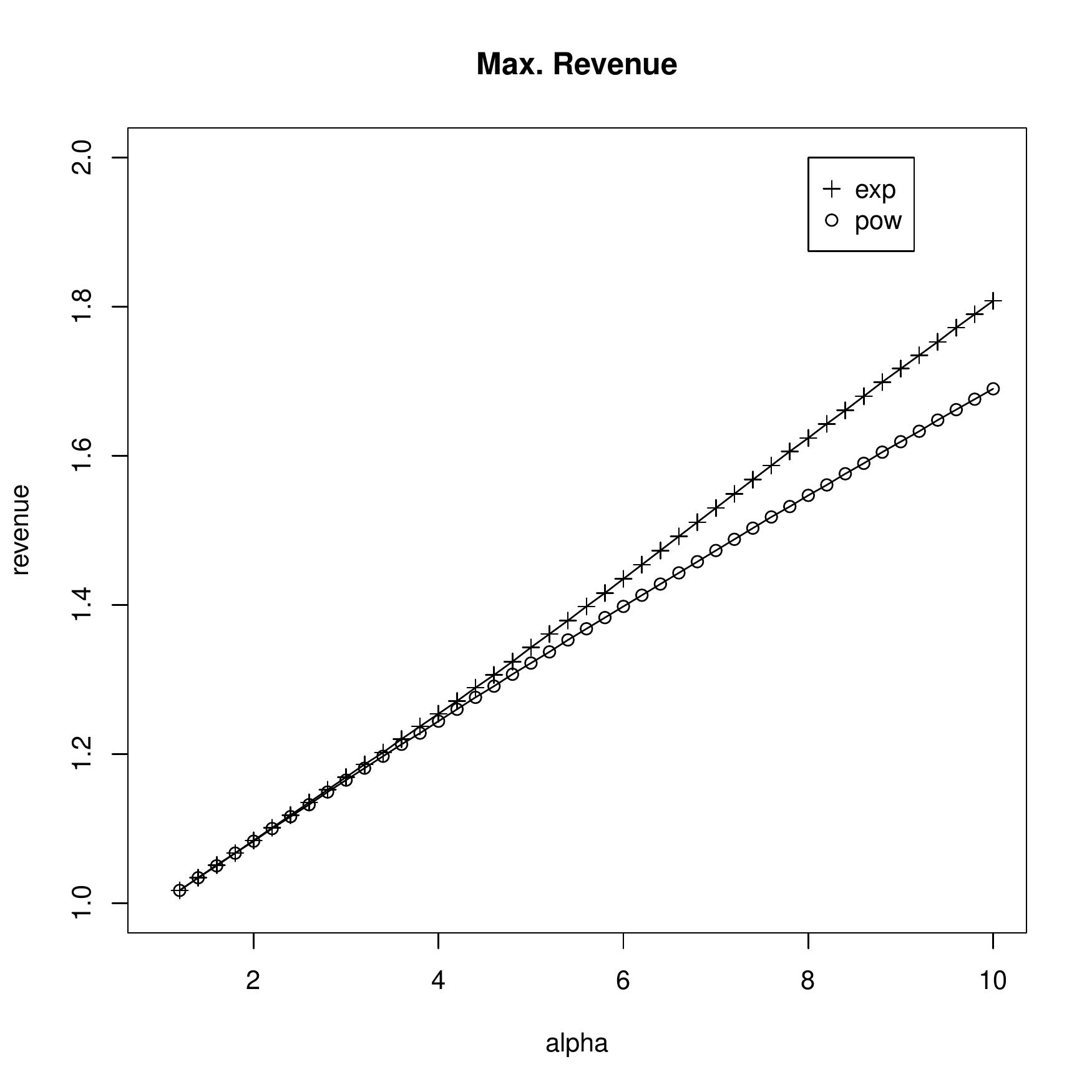} \label{sa_rev}}
\subfigure[]{\includegraphics[width = 0.48\textwidth]{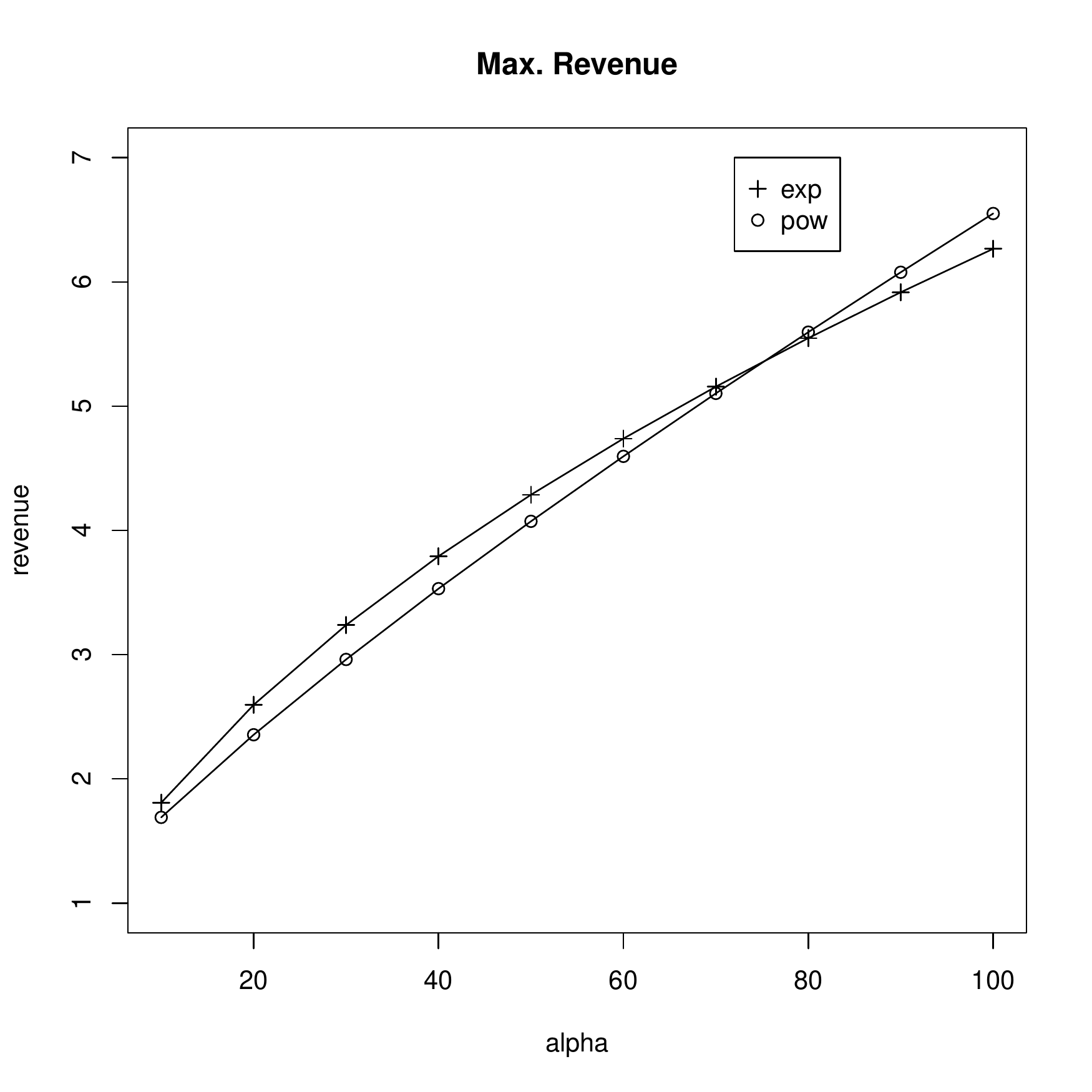} \label{ma_rev}}
\subfigure[]{\includegraphics[width = 0.48\textwidth]{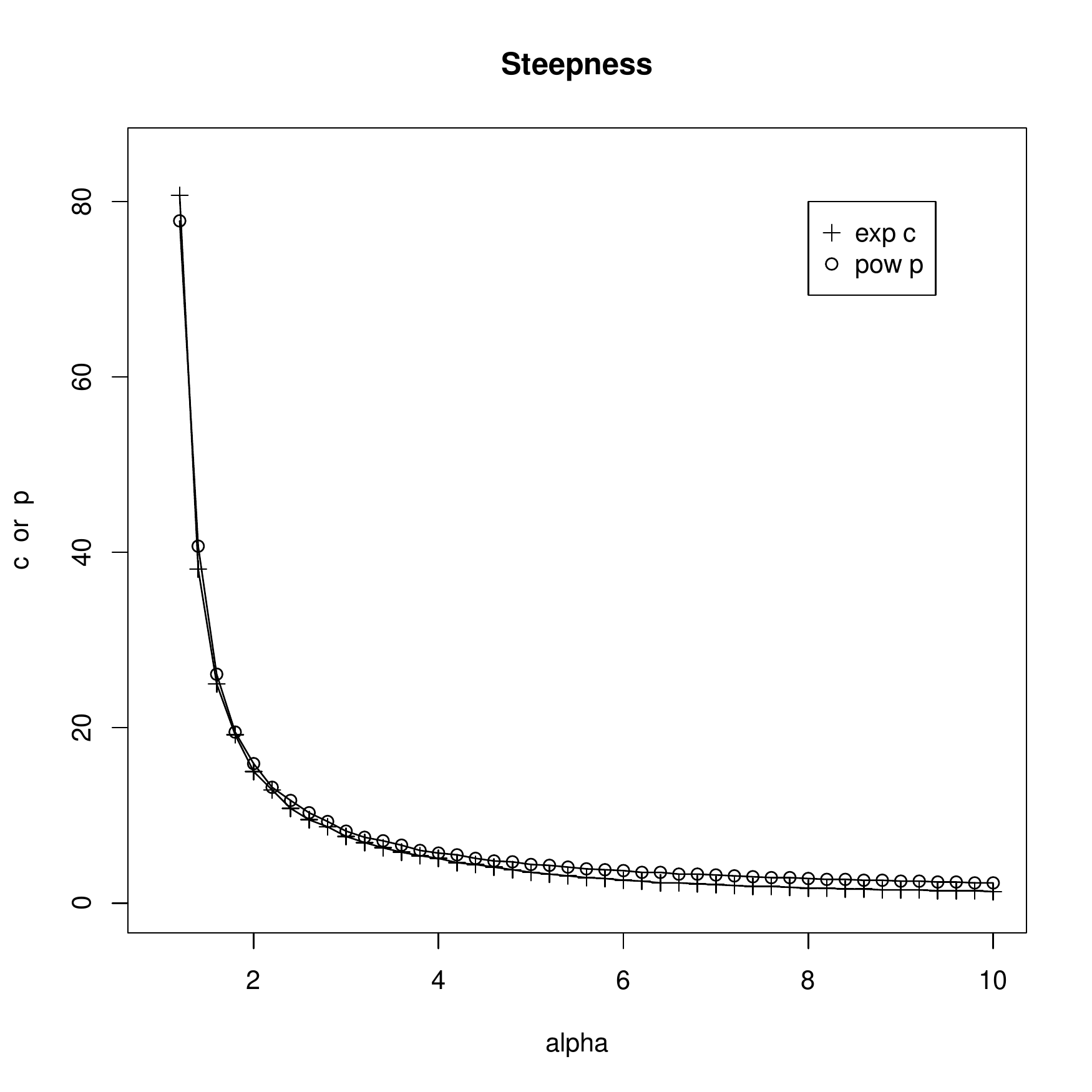} \label{sa_steep}}
\subfigure[]{\includegraphics[width = 0.48\textwidth]{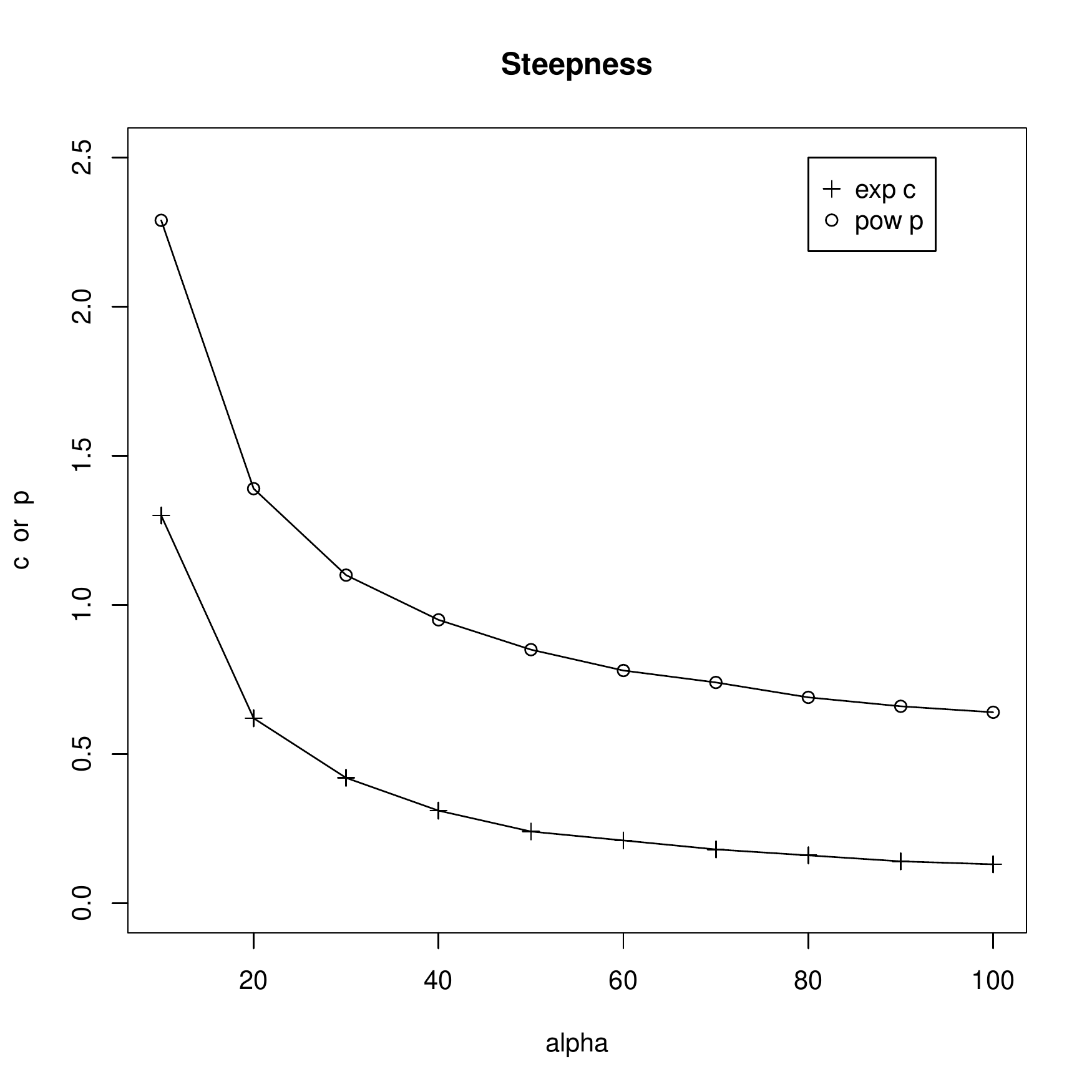} \label{ma_steep}}
\caption{Maximum Revenue and Maximimizing Steepness Parameters for $n = 2$.}
\label{rev_steep}
\end{figure}

\begin{figure}[ht]
\subfigure[]{\includegraphics[width = 0.48\textwidth]{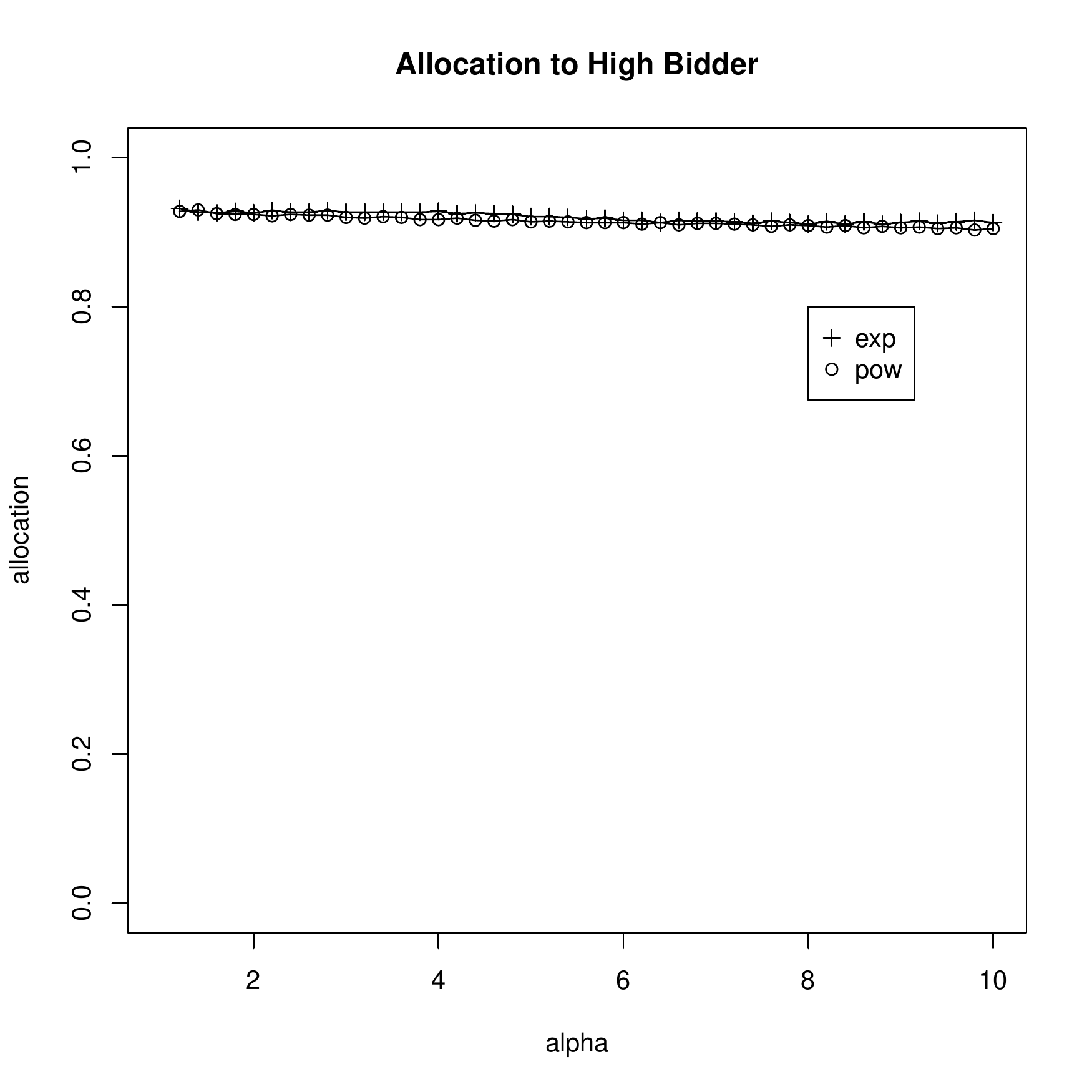} \label{sa_alloc}}
\subfigure[]{\includegraphics[width = 0.48\textwidth]{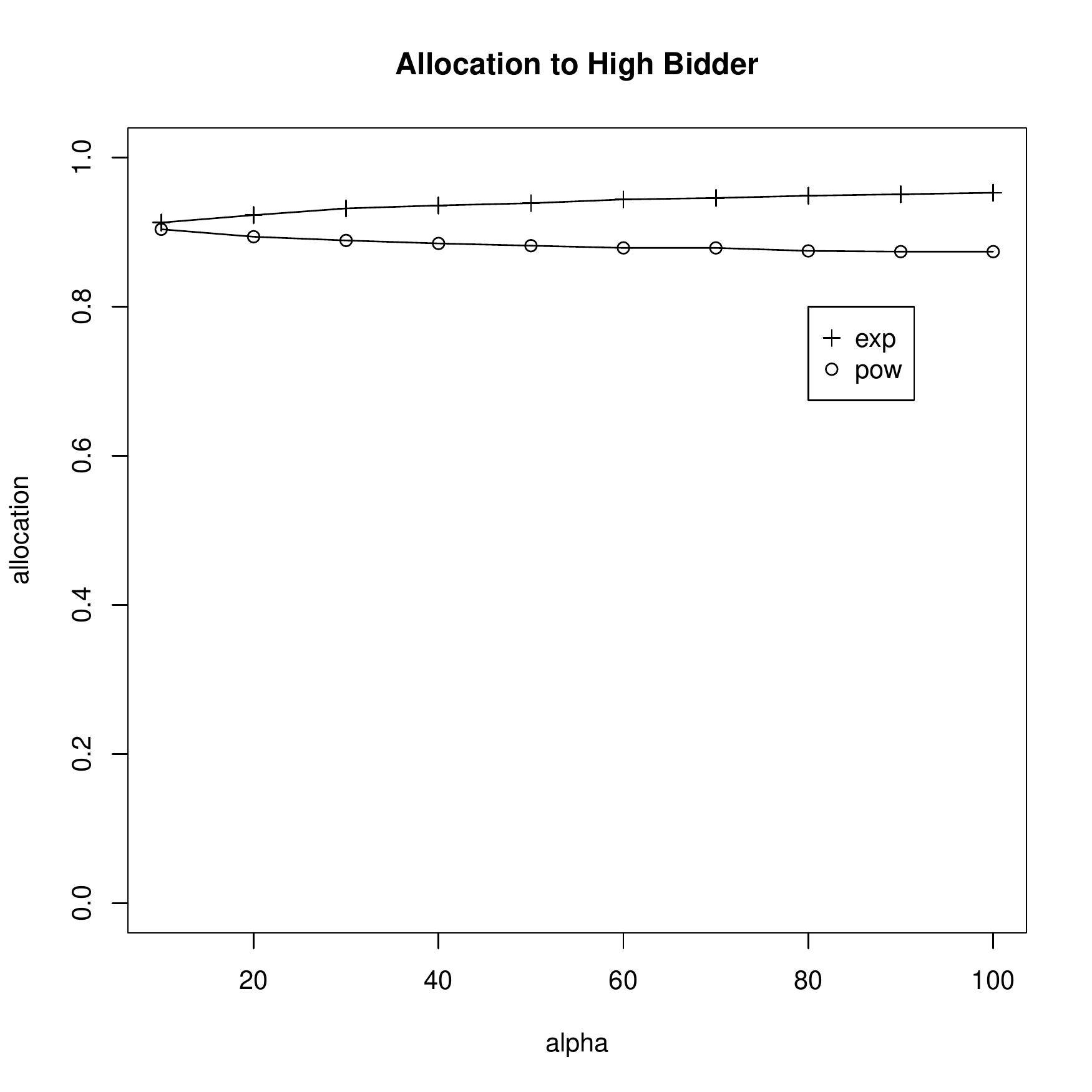} \label{ma_alloc}}
\subfigure[]{\includegraphics[width = 0.48\textwidth]{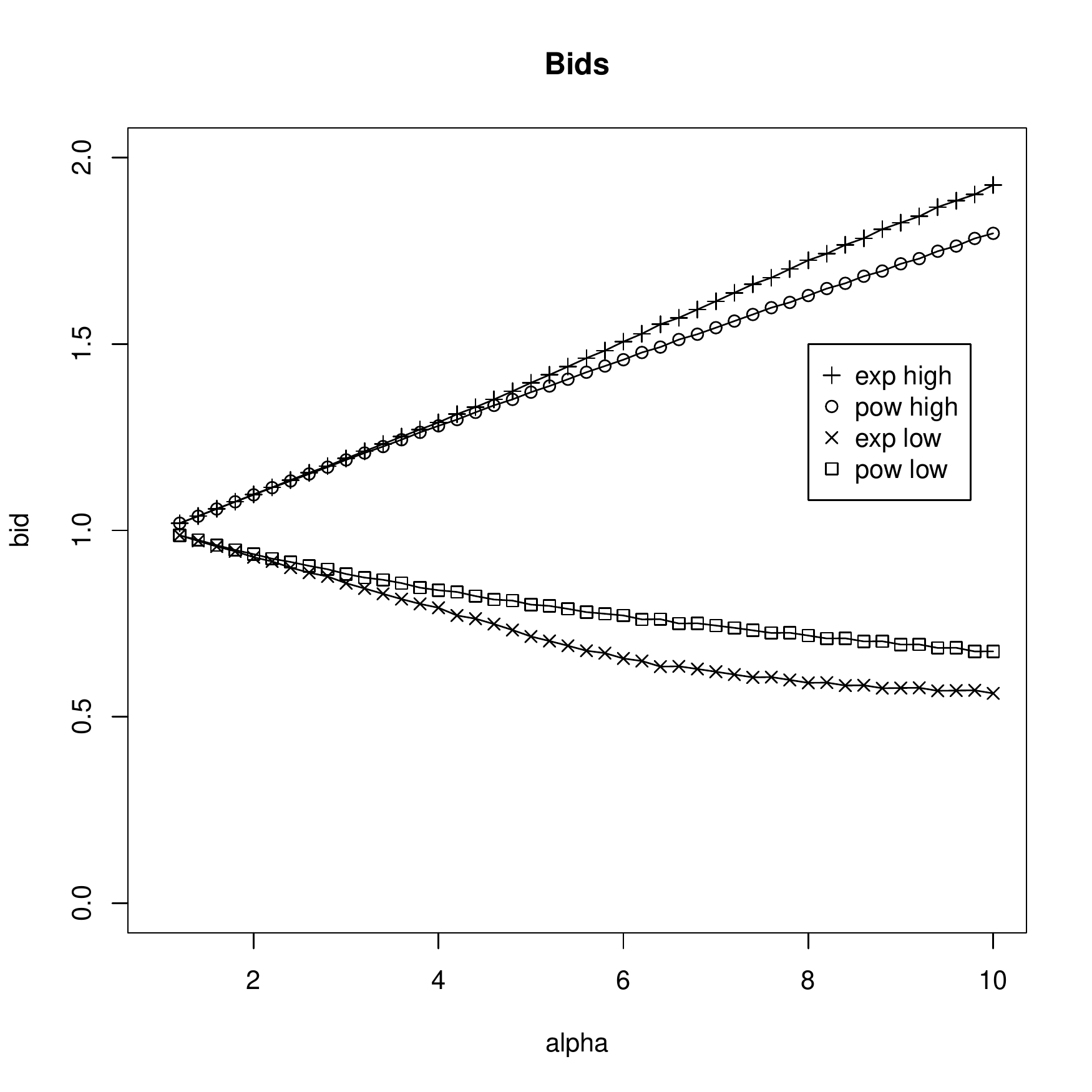} \label{sa_bids}}
\subfigure[]{\includegraphics[width = 0.48\textwidth]{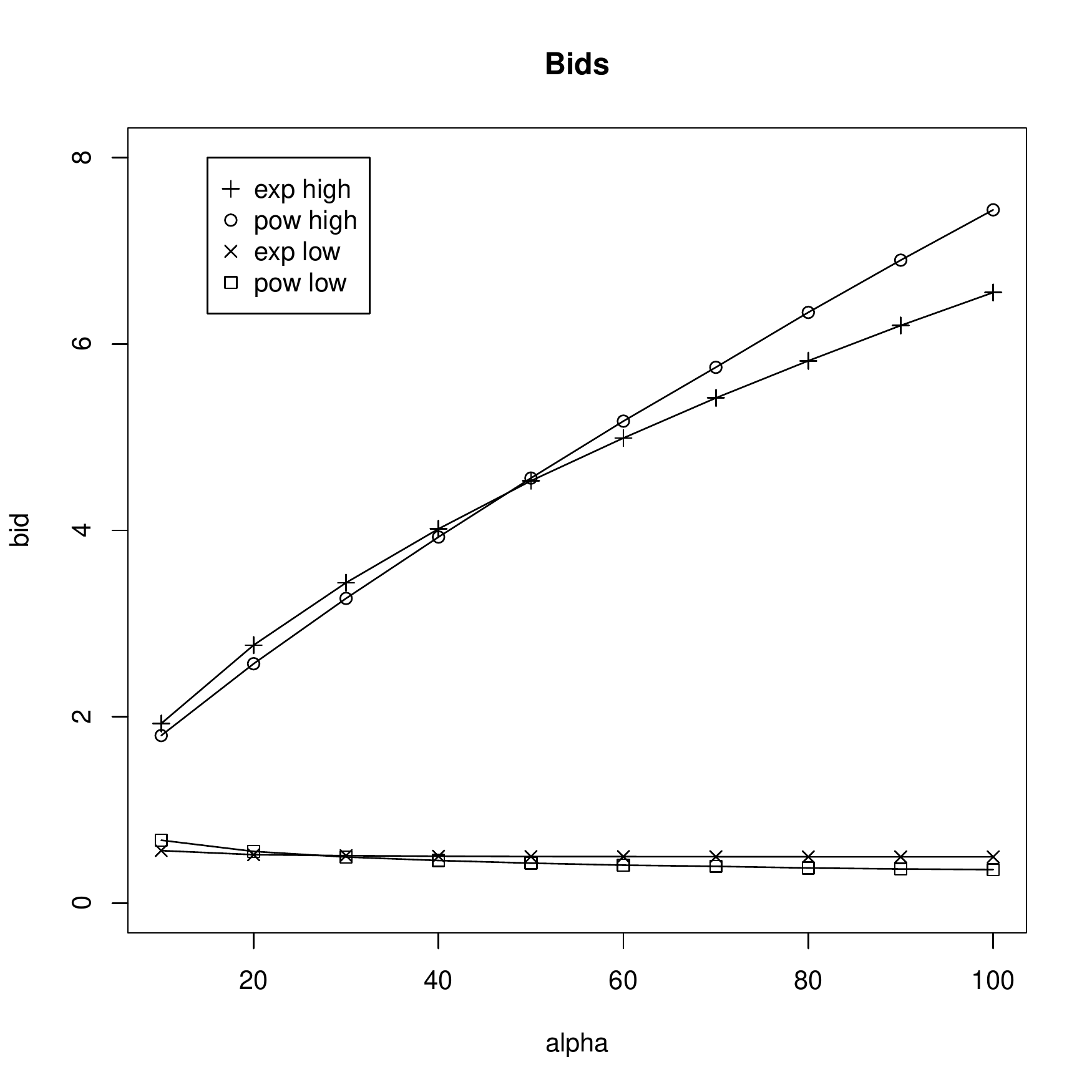} \label{ma_bids}}
\caption{Allocation to High Bidder and Bids at Maximum Revenue for $n = 2$.}
\label{alloc_bids}
\end{figure}

\begin{figure}[ht]
\subfigure[]{\includegraphics[width = 0.48\textwidth]{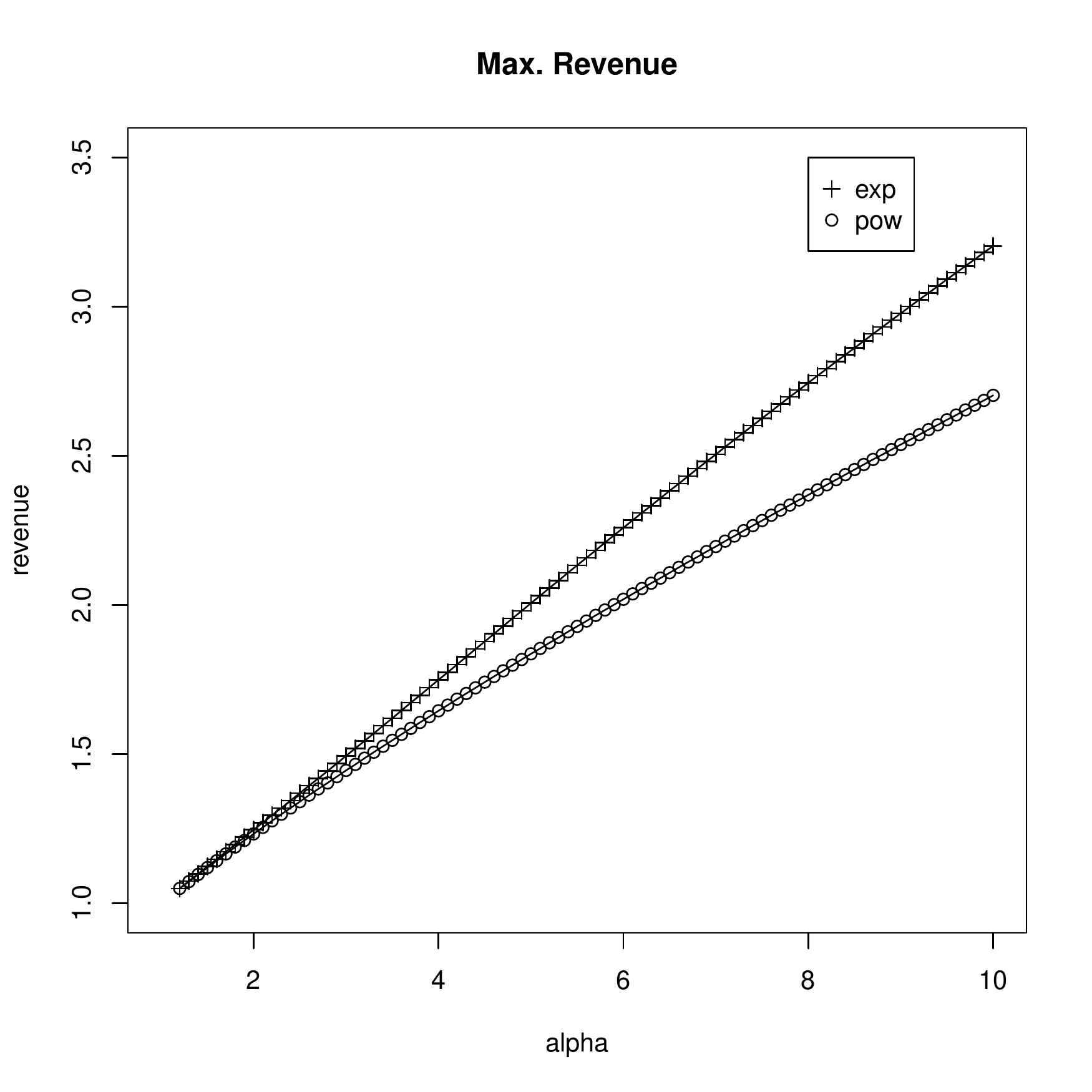} \label{osa_rev}}
\subfigure[]{\includegraphics[width = 0.48\textwidth]{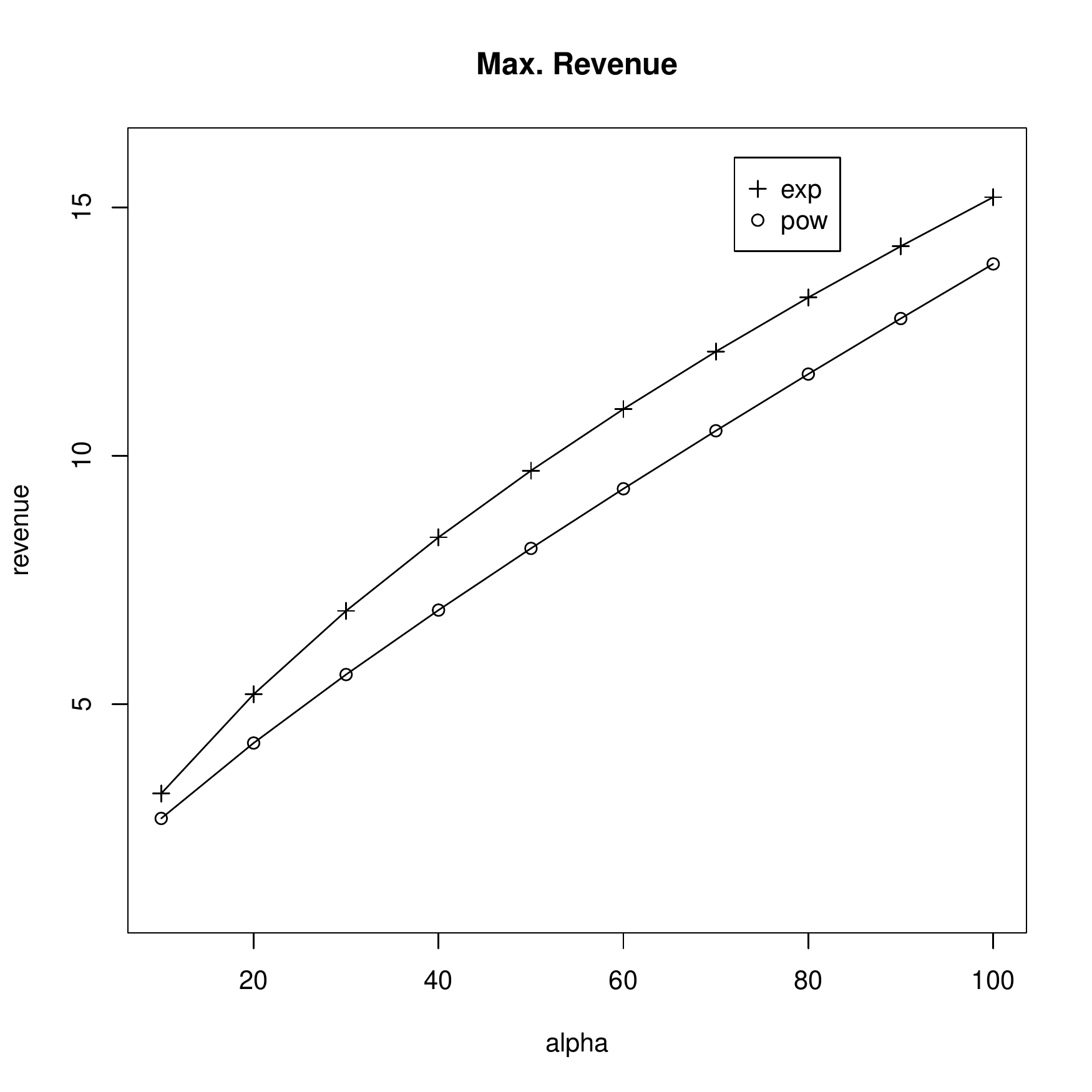} \label{oma_rev}}
\subfigure[]{\includegraphics[width = 0.48\textwidth]{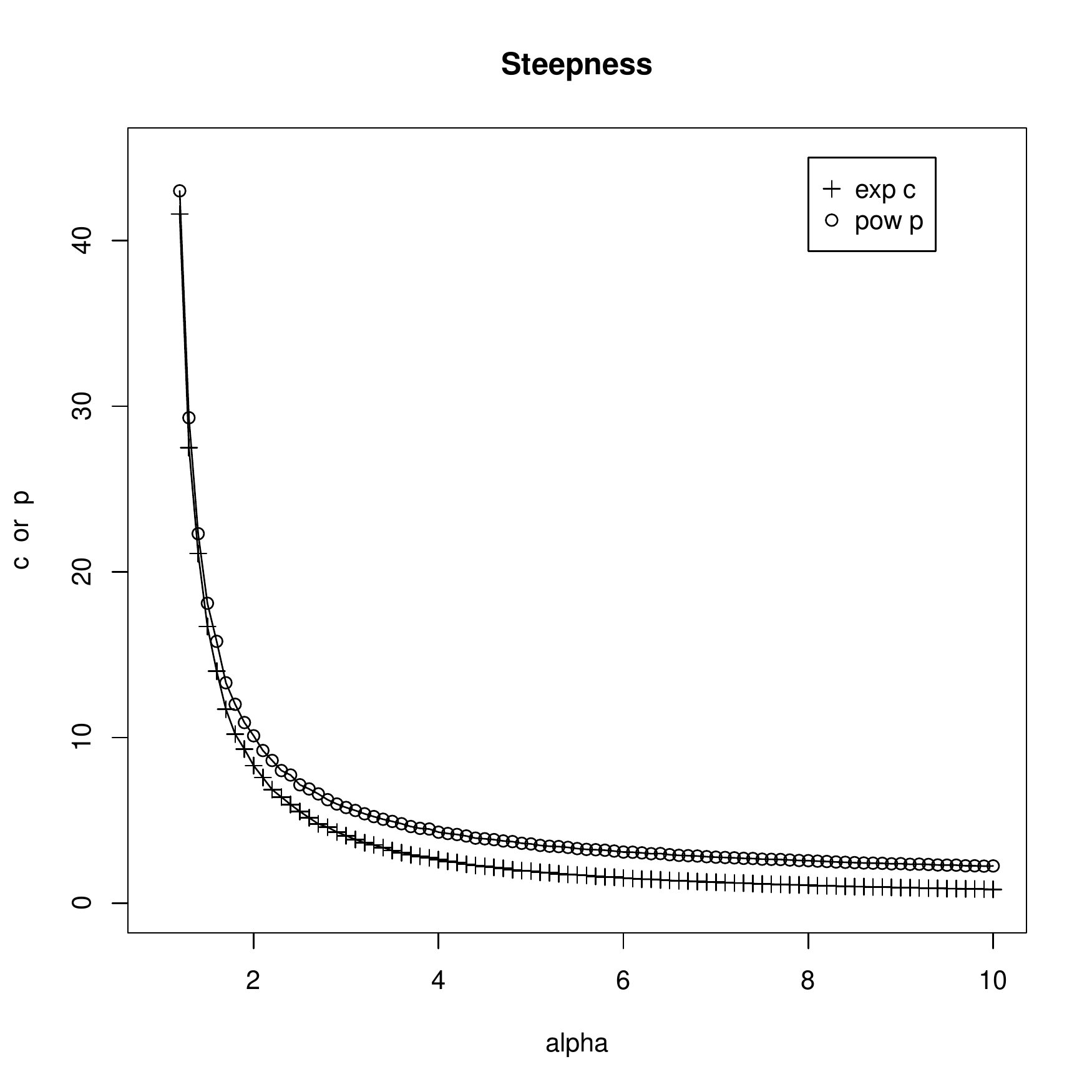} \label{osa_steep}}
\subfigure[]{\includegraphics[width = 0.48\textwidth]{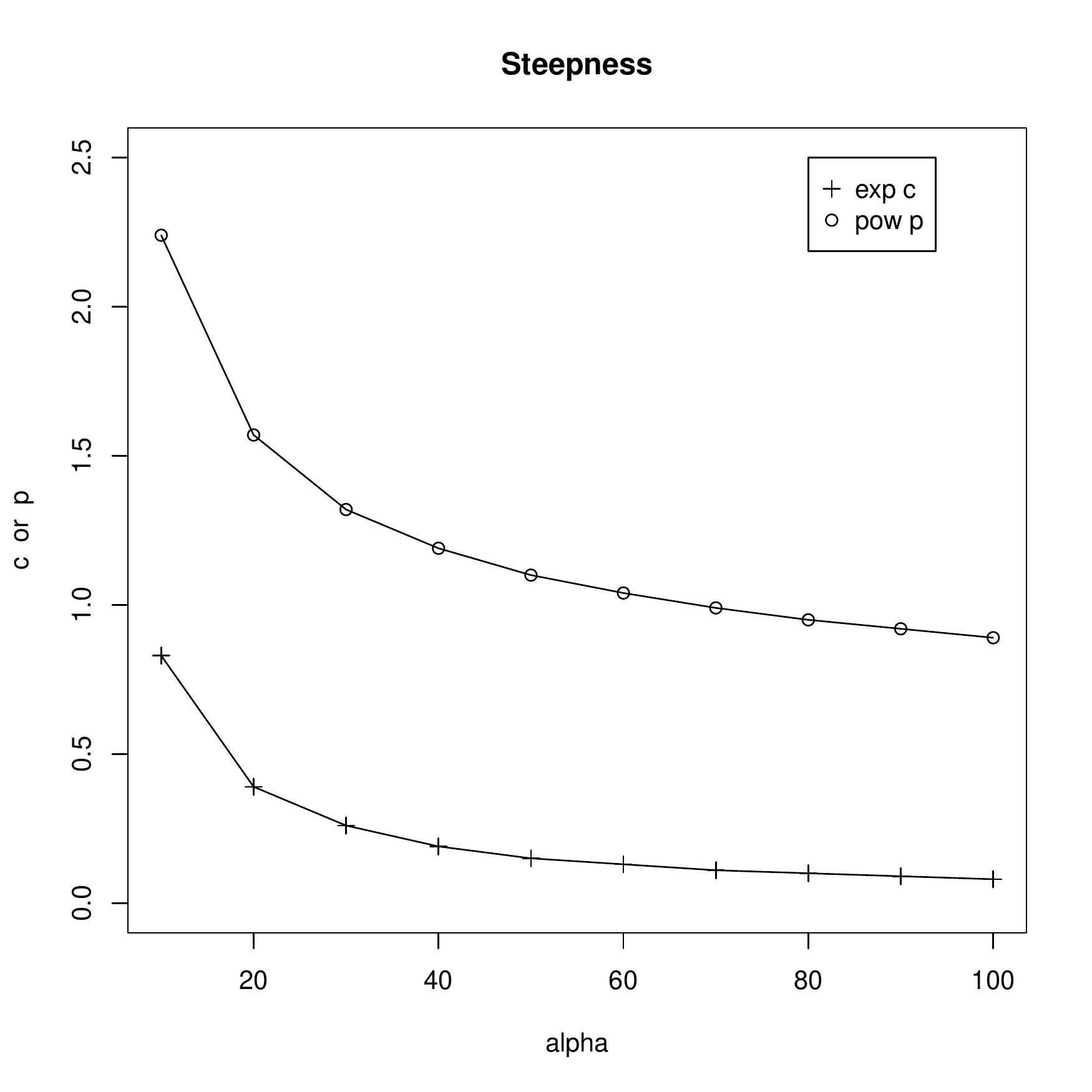} \label{oma_steep}}
\caption{Maximum Revenue and Maximimizing Steepness Parameters for $n = 10$.}
\label{olos_rev_steep}
\end{figure}

\begin{figure}[ht]
\subfigure[]{\includegraphics[width = 0.48\textwidth]{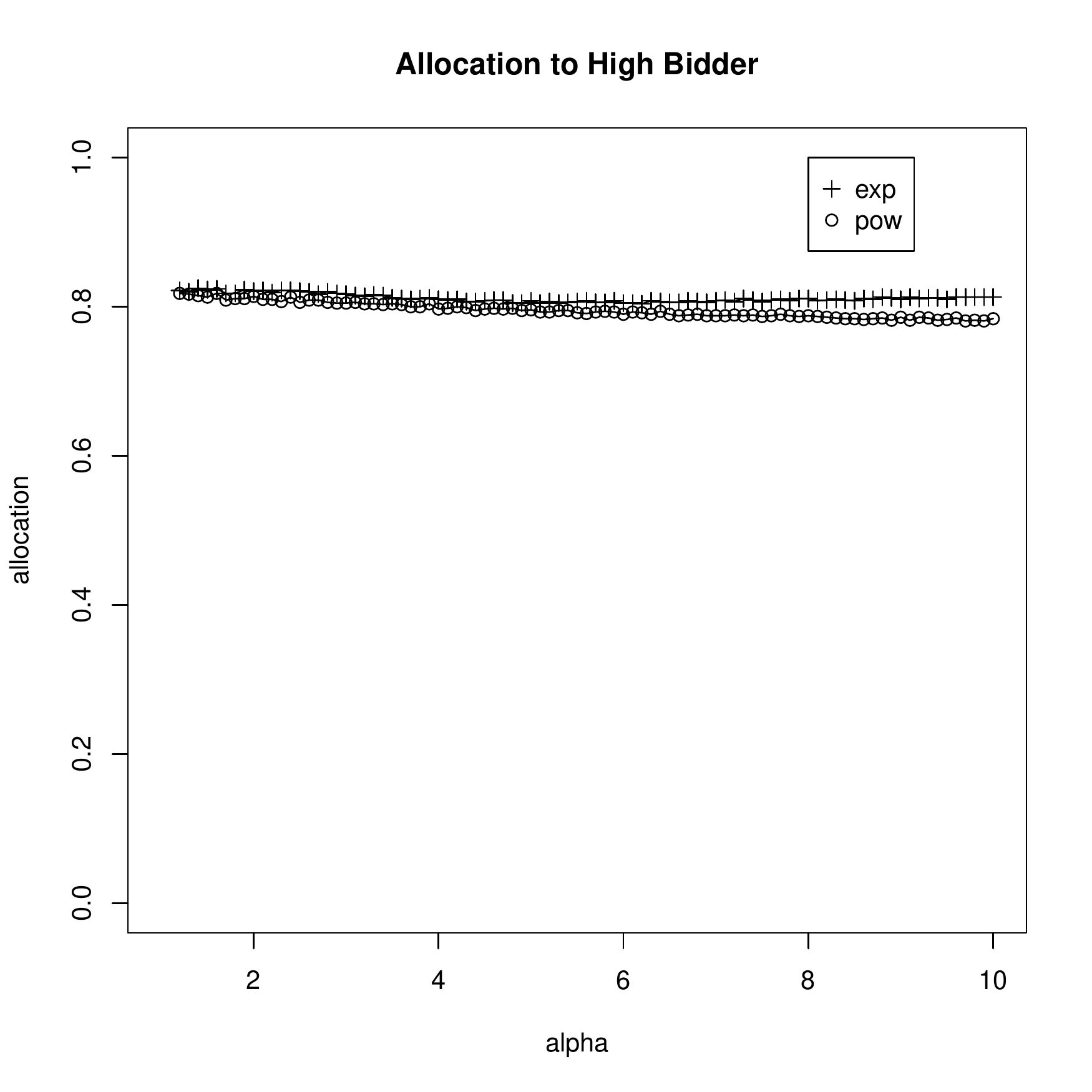} \label{osa_alloc}}
\subfigure[]{\includegraphics[width = 0.48\textwidth]{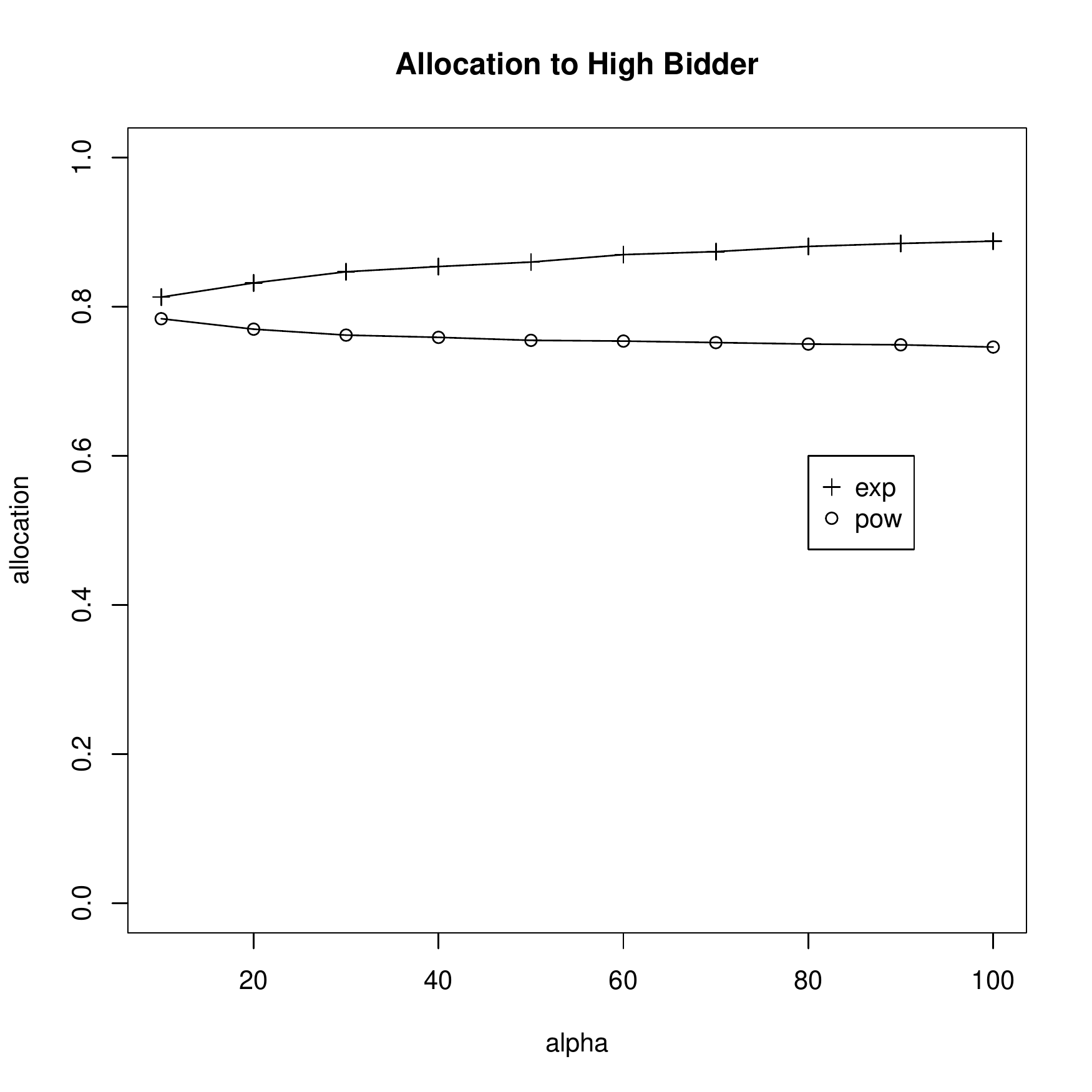} \label{oma_alloc}}
\subfigure[]{\includegraphics[width = 0.48\textwidth]{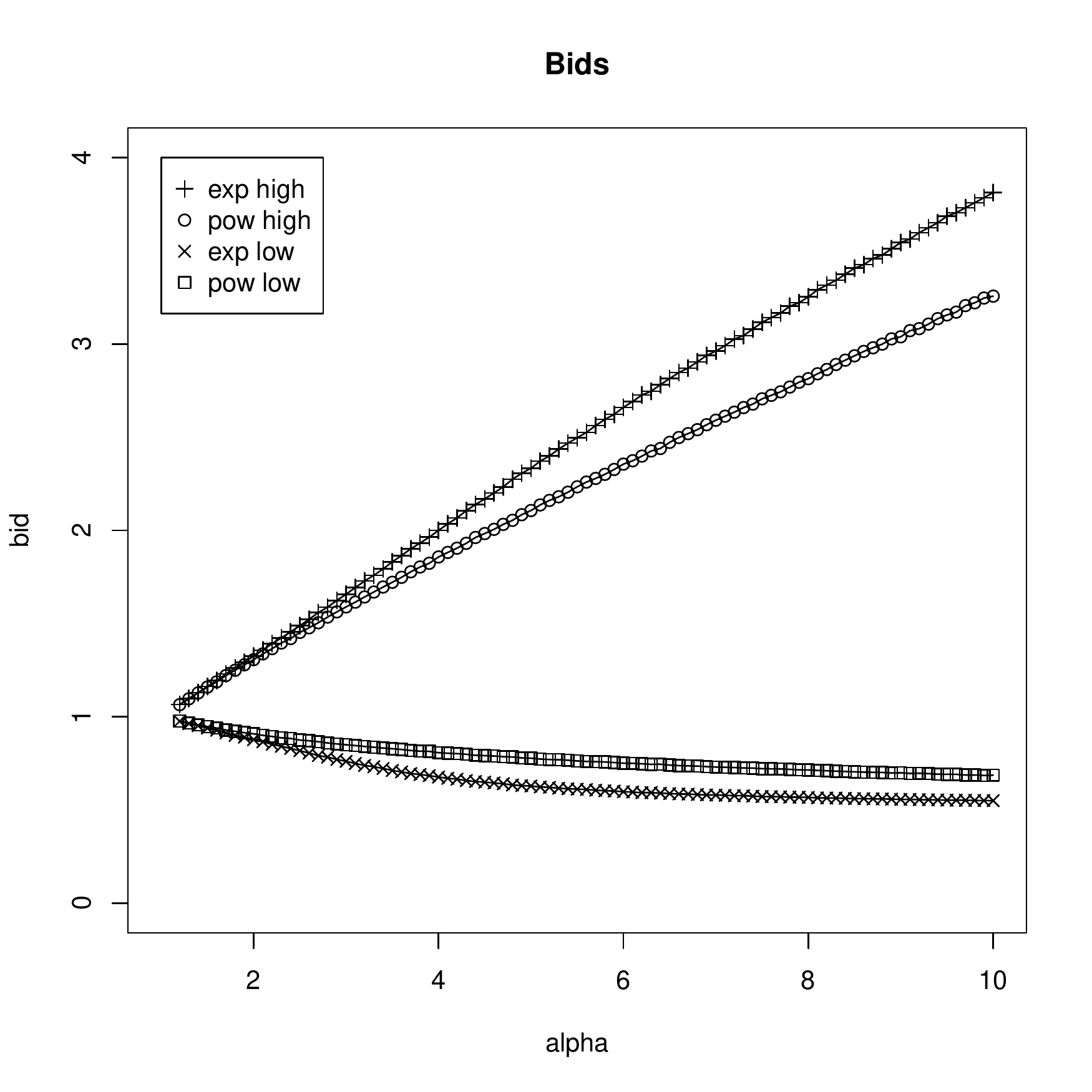} \label{osa_bids}}
\subfigure[]{\includegraphics[width = 0.48\textwidth]{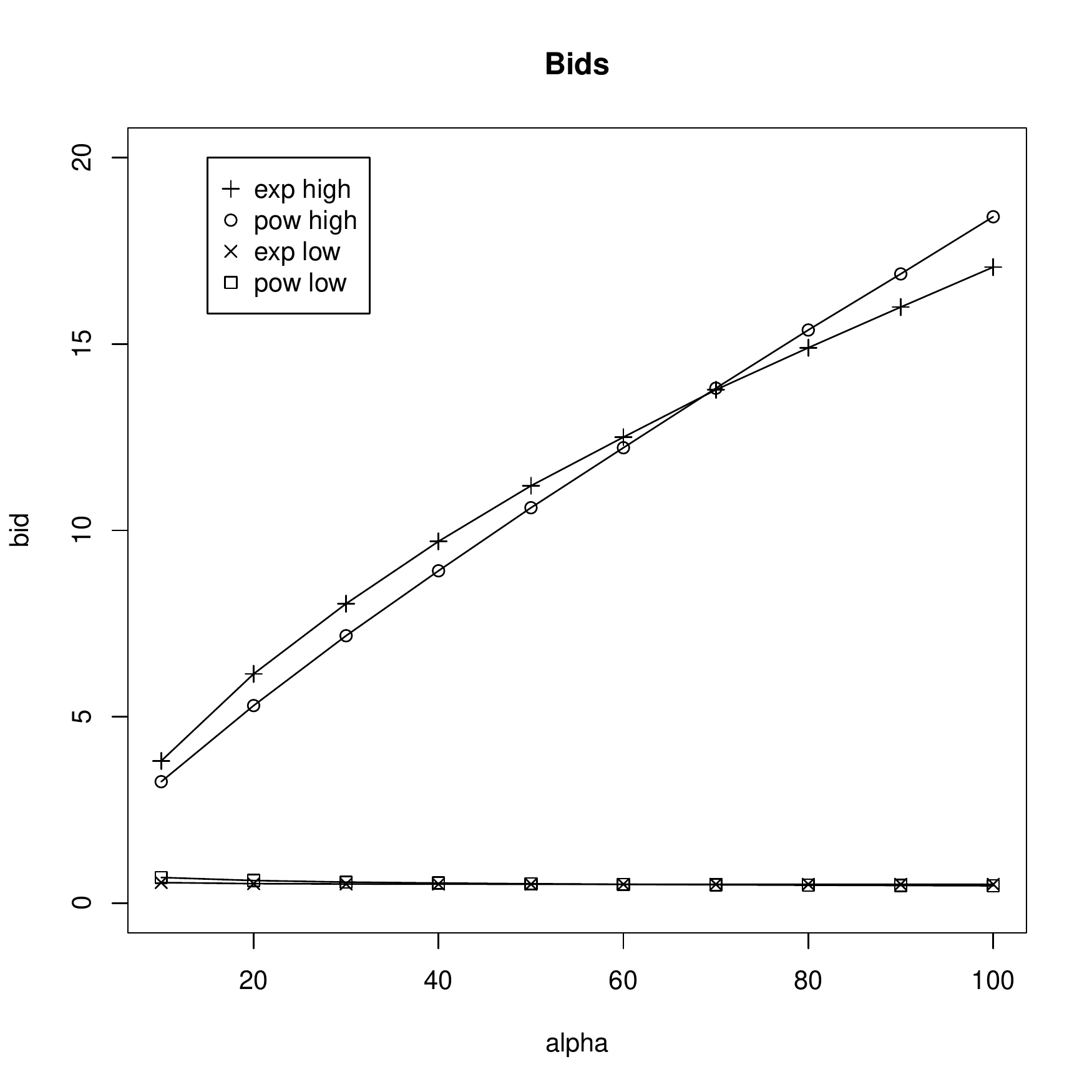} \label{oma_bids}}
\caption{Allocation to High Bidder and Bids at Maximum Revenue for $n = 10$.}
\label{olos_alloc_bids}
\end{figure}

\section{Discussion}
For exponential weight functions, we have characterized bids at a pure-strategy Nash equilibrium. We have also shown that exponential weight functions can provide more revenue than power weight functions, for a two-bidder auction and a ten-bidder auction with a single high private value bidder and nine equal private value bidders, if there is strong or even moderate competition. So exponential weight functions should be considered for quasi-proportional auctions. 

For both exponential and power weight functions, increasing competition by having a ratio of bidder private values closer to one increases the steepness of the revenue-maximizing weight function. The relationship between this steepness and increasing competition by having more bidders is more complex. Analyzing and understanding this relationship warrants further study. Wen et al. \cite{wen15} show that, for power weight functions, as the number of bidders increases from two, the optimal steepness initially decreases, then increases. As $\alpha$ increases, the initial decrease becomes less pronounced. Comparing Figures \ref{sa_steep} and \ref{osa_steep} in this paper, for power weight functions, the optimal steepness for ten bidders is actually less than for two bidders until $\alpha=10$. Then, comparing Figures \ref{ma_steep} and \ref{oma_steep}, there is a crossover before $\alpha=20$. For exponential weight functions, this effect is even more pronounced: the optimal steepness for ten bidders is less than for two bidders for $\alpha = 1.2$ to $\alpha = 100$ and beyond (based on the same pairs of figures and on computations to $\alpha=1000$, not shown in the figures). Understanding this phenomenon may give insight on how to develop new weight functions and how to select a weight function. 


One direction for future research is to develop effective methods to choose the steepness parameter, given limited information about bidders' private values. Another direction is to investigate more complex weight functions, for example weight functions that are relatively steep for low bids and less steep for high bids. For a related idea, refer to \cite{nguyen10}, for a quasi-proportional auction that determines allocations based on the ratio between bids. Another direction for future research is to investigate analogs to reserve prices for quasi-proportional allocations, for example have the auctioneer submit a bid and withhold the allocation for that bid from the allocation to bidders. This method decreases risk to the seller compared to using a reserve price, because this method makes zero revenue less likely. 

One more direction for future research is to explore revenue maximization over a more general set of weight functions that includes both powers and exponentials. Since $ e^x = 1 + x + \frac{x^2}{2!} + \frac{x^3}{3!} + \ldots $:
$$ e^{cx}-1 = cx + \frac{c^2}{2!} x^2 + \frac{c^3}{3!} x^3 + \ldots $$
So our exponential weight functions belong to the general set of polynomial weight functions (without constant terms, so $f(0)=0$), if we allow infinite degree:
$$ f_{\cc}(x) = \sum_{i=1}^{\infty} c_i x^i, $$
where $\cc = (c_1, c_2, \ldots )$ is the vector of coefficients. Of course, power weight functions also belong to this set, as polynomials with a single coefficient of one in $\cc$ and the other coefficients zero. It would be interesting to understand how the coefficients of the revenue-optimizing weight function of the form $f_{\cc}$ depend on the relationships between private values and on the number of bidders.

\bibliographystyle{abbrv}
\bibliography{bax}

\end{document}